\documentclass[aps,twocolumn,noshowpacs, showkeys, superscriptaddress, floatfix,nobalancelastpage, notitlepage]{revtex4-1}

\usepackage{graphicx}% Include figure files
\usepackage{dcolumn}% Align table columns on decimal point
\usepackage{bm}% bold math
\usepackage{bbm}
\usepackage{quantikz}
\usepackage{amsmath,amssymb} 
\usepackage{amsthm,epsfig,tikz}
\usepackage[justification=justified]{caption}
\usepackage{xcolor}
\usepackage{verbatim}
\usepackage[normalem]{ulem}
\usepackage[rightcaption]{sidecap}
\usepackage{subcaption}
\usepackage{forest,adjustbox}
\usepackage{multirow}
\usepackage{hyperref}  
\hypersetup{colorlinks,citecolor=black,urlcolor=black,linkcolor =black,hypertexnames=true,pdfproducer={Me}}
\usepackage{tikz}

\usepackage{float}
\makeatletter
\let\newfloat\newfloat@ltx
\makeatother
\usepackage{algorithm}
\usepackage{algpseudocode}

\algrenewcommand{\algorithmiccomment}[1]{\hskip2pt\textbackslash\textbackslash \ #1}

\def\bra#1{\mathinner{\langle{#1}|}}
\def\ket#1{\mathinner{|{#1}\rangle}}

\renewcommand{\part}[2]{\frac{\partial #1}{\partial #2}}

\newcommand{\minus}{  \scalebox{0.45}[1.0]{\( - \)}  }

\newtheorem{re}{Remark}
\newtheorem{prop}{Proposition}
\newtheorem{claim}{Claim}

\begin{document}

\title{Scoring Anomalous Vertices Through Quantum Walks}

\author{Andrew Vlasic}
\email{avlasic@deloitte.com}
\affiliation{Deloitte Consulting, LLP, Chicago, IL}

\author{Anh Pham}
\email{anhdpham@deloitte.com}
\affiliation{Deloitte Consulting, LLP, Atlanta, GA}

\begin{abstract}
With the constant flow of data from vast sources over the past decades, a plethora of advanced analytical techniques have been developed to extract relevant information from different data types ranging from labeled data, quasi-labeled data, and data with no labels known a priori. For data with at best quasi-labels, graphs are a natural representation of these data types and have important applications in many industries and scientific disciplines. Specifically, for unlabeled data, anomaly detection on graphs is a method to determine which data points do not posses the latent characteristics that is present in most other data. There have been a variety of classical methods to compute an anomalous score for the individual vertices of a respected graph, such as checking the local topology of a node,random walks, and complex neural networks. Leveraging the structure of the graph, we propose a first quantum algorithm to calculate the anomaly score of each node by continuously traversing the graph with a uniform starting position of all nodes. The proposed algorithm incorporates well-known characteristics of quantum random walks, and, taking into consideration the NISQ era and subsequent ISQ era, an adjustment to the algorithm is given to mitigate the increasing depth of the circuit. This algorithm is rigorously shown to converge to the expected probability, with respect to the initial condition. 
\end{abstract}

\date{\today}

\maketitle

\section{Introduction}\label{sec:intro}
We propose a simple, yet, elegant methods to calculate an anomaly score for the individual nodes within a graph structure.  %how anomalous a node is with respect to the structure of the graph%. 
To the author's knowledge, this is the first quantum based algorithm to output a score for the individual nodes within a graph, yielding further insight into the structure of the graph, and the relationship between different nodes. For our experiments, we considered weighted undirected and directed graphs without self-loops. The techniques are rigorously sound, using the basis of mixing, established by Aharonov et al. \cite{aharonov2001quantum} as the quantum counterpart to Markovian random walks, to exhibit convergence of the algorithm. Furthermore, to mitigate the depth of the quantum circuit, we derive an algorithm that modularizes the number of iterations of the walk operator, which will be demonstrated in one of our examples to demonstrate the reduction of circuit depth using our modified version of the regular continuous time random walk (CTQW). This algorithm is rigorously shown to converge to the same state is the algorithm without modularization. An explicit method is given for graphs whose adjacency matrix is not Hermitian, and we discuss how to incorporate discrete time quantum walks of weighted graphs and display a technique to adjust for case when the walk operator is Hermitian but not unitary.

The general meaning of an anomalous node is a node which significantly differs compared to an expected node, as to arouse suspicion, given in Hawkins \cite{hawkins1980identification}. Observe that the `expected node' is particular to the context of the data and the nature of dynamic that generated the data. Essentially, anomaly detection algorithms find outliers with distributions of structure and attributes \cite{pourhabibi2020fraud}, which may include the local topology of a node, characteristics of the edges, characteristics of the nodes, and characteristics of the identified subgraphs. These algorithms either explicitly leverage  the distribution of the data, such as the (potentially) invariant distributions of random walks \cite{wang2018new} or GANs \cite{chen2020generative,guo2023regraphgan}, or implicitly leverage the distribution of the data, like matrix decomposition \cite{sun2007less,peng2018anomalous}, and combinations of machine learning techniques which include vector representations of nodes \cite{grover2016node2vec}, graph attention networks \cite{velickovic2017graph}, and autoencoders \cite{du2022graph}. A more detailed survey can also be found in Xing et al. \cite{xing2024survey}. 

For this paper, we take the perspective of random walks on graphs \cite{moonesinghe2006outlier,wang2009discovering,wang2018new} consider a node to be anomalous (or an outlier) if it is not traversed as often with respect to other nodes. From this definition the structure of the graph, and not individual characteristics of nodes (such as the number of edges or triangles formed), dictates whether a vertex is an outlier. For instance, a node with a relatively small edge count will not be considered anomalous if the majority of the edges are connected to well-connected vertices as it will be visited often. 

The algorithm is based on a result in mixing that states the average of the summed average of each sequential run converges to a probability with respect to the initial condition. This means that for $n$ steps, at the $i^{th}$ step the walk in implemented $i$ times in the circuit and the output of the shots averaged and summed to the $(i-1)^{th}$ step. After a respective step, say $i_0$, the depth of the circuit is too unwieldly to implement in a general quantum processing unit (QPU) in the NISQ era and subsequent ISQ era. We adjust for the depth of the quantum circuit by implementing a run up to $i_0$ implementations of the walk operator then feed the output vector back into the circuit until the number of iterations of the walk operator sums to the given number. For instance, if $i_0=3$ and the step is 7, then the circuit is first ran with three walk operators, then output is fed back into the circuit and ran with three walk operators, and finally, the output vector is fed back into the circuit and ran with one walk operator. Using this process, we can take advantage of other data loading technique such as Grover-Rudolph \cite{grover2002creating} to load the output states for the next time step to reduce the accumulated error when walks are implemented on real hardware. In addition, the number of qubits scales logarithmically with the number of columns of the adjacency matrix, thus making this algorithm scalable for large graph. 

Within the NISQ era, continuous time quantum walks have been explored on hardware devices for simple graphs and hypercube within a few qubits with comparable results obtained from the simulator \cite{MADHU2023e13416}. In addition, noise mitigation techniques \cite{Endo_2019} have been explored for Hamiltonian simulation which is a critical step to perform the time evolution of the quantum state with the walkers. 

%many noise mitigation techniques introduce quantum operators into the circuit, for instance, including zero noise extrapolation \cite{kandala2019error,majumdar2023best} and probabilistic error cancellation \cite{van2023probabilistic,mcdonough2022automated}, and the proposed algorithm may incorporate the extra gates without sacrificing accuracy.}

While intuitive, we show this algorithm converges to the same probability as the algorithm without this modularity. In addition, given the scale and fidelity of a particular QPU, one may implement the algorithm without intermediate steps, as is displayed in the quantum walk literature on mixing. 

The structure of the paper is organized as follows. Section \ref{sec:classical} gives a brief overview of graphs, the representation of graphs as matrices, and classical random walks. Section \ref{sec:cont} gives a background of continuous time quantum walks, results from the concept of mixing, and discusses a method to implement a continuous time quantum walk on a directional graph with an asymmetric adjacency matrix. Section \ref{sec:sym-score} derives the algorithm, gives a rigorous basis for convergence, and argues why the algorithm utilizes the adjacency matrix. A small change to the algorithm is shown when the adjacency matrix is not Hermitian. For completeness, Section \ref{sec:coin} gives a brief overview of discrete quantum walks and an explicit walk operator is identified for weighted graphs. When the discrete walk operator is Hermitian but not unitary, an explicit transformation of the walk operator and how to extract information is discussed. Thus, giving a foundation to implement the algorithm with discrete quantum walks. Finally, Section \ref{sec:diff-ham} shows the results of the algorithm ran on toy graphs with symmetric and asymmetric adjacency matrices, and with different Hamiltonian generator matrices. Completing the comparison, a classical random walk is also calculated. Due to the scaling of the simulated qubits, only toy graphs were considered. Section \ref{sec:discussion} summarizes the paper and delves into potential future research.   

\section{Classical Random Walks}\label{sec:classical}
In this section we give an overview of representations of a graph, randomly walking a graph, and calculating anomaly scores for every node. The overview will establish a foundation when quantum algorithms are introduced and explored. 

A graph $G$ is denoted as the pair $(V,E)$ where $V$ is the set of vertices and $E$ is the set edges between the connected vertices. Given that whether two nodes are connected is binary, $E$ can be represented as a matrix. For $|V|=n$, each node is given a unique label $1,\ldots,n$. Then we can construct a matrix $A$ that represents $E$ where entry $a_{ij}= \left\{\begin{array}{cc}
   1  & \mbox{if } e(i,j)\in E \\
   0  & \mbox{otherwise}
\end{array} \right. .$  The matrix $A$ is called the \textbf{adjacency} matrix. Note that the labeling of the nodes does not affect the dynamic since there exists a permutation matrix $\mathcal{P}$ such that $A = \mathcal{P}^T \Tilde{A} \mathcal{P}$.

Collecting information about the immediate topology of each node, define $D$ as a diagonal matrix where entry $d_{ii} = \sum_{j} a_{ij}$. With this construction, one may see how to construct these matrices with a graph where each edge has a respective weight; each weight is assumed to be strictly positive. For more in-depth information on graphs and matrices, see \cite{bapat2010graphs}.

With the matrix representation, we may mathematically describe how information (or energy) flows through a graph. Given an orientation of a flow, define the \textbf{incidence} matrix $Q$ where $q_{ij}=1$ if $v_i$ is the 'positive end', $q_{ij}=-1$ if $v_i$ is the 'negative end', and $0$ if an edge does not exist. With the incidence matrix we have a form of the discrete Laplace operator, or Laplacian, $L=QQ^T$. Interestingly, the Laplacian is independent of the orientation, and
\begin{equation}\label{eq:lapc}
    L = D - A .
\end{equation}
See Merris \cite{merris1994laplacian} for more information.

Traversing a graph with a random walk has been studied for decades with many advancements, see \cite{lovasz1993random,aldous2002reversible,blanchard2011random,riascos2018random,riascos2021random}. However, for the purposes of this paper, only the Markov random walk with the derivation from the adjacency matrix will be covered; a random variable $\{x_t\}_{t \in I}$ has a Markov distribution if $P(x_n=j | x_{n-1}=i_{n-1},  \ldots, x_0 = i_0) = P(x_n=j | x_{n-1}=i_{n-1})$. 

Particularly, for a global traversing denote $M = (p_{ij})_{i,j \in V}$ as the matrix of probabilities where  $p_{ij}= a_{ij}/d_{ii}$. Ergo, $M=AD^{-1}$, and for current position $P_t$ the transition is $P_{t+1}=M^T P_t$. From the structure of graphs we get the existence of a stationary distribution $\pi$ such that $\pi M = \pi$. Furthermore, for any initial condition $P_0$, $\displaystyle \lim_{t \to \infty} P_t = \pi$, which makes the distribution $\pi$ invariant. The invariant measure may be approximated with the formula $\displaystyle \pi^{t+1} = \frac{ P^T \pi^{t} }{ || P^T \pi^{t} ||_1}$, with $||\cdot||_1$ as the $L^1$ norm.

If information is known a priori about the nodes, there is an augmented Markov random walk, called the maximal entropy random walk (MERW), where bias on the nodes is introduced. For the construction, take the normalized eigenvector of $A$, say $\vec{\xi}$, where for $A \vec{\xi} = \chi \vec{\xi}$, $\chi$ is the largest eigenvalue. Define $\Tilde{A}$ as the augmented adjacency matrix where define $\Tilde{a}_{ij} = a_{ij}\xi_{i}\xi_{j}$. Notice that 
\begin{equation}\label{eq:mea}
    \Tilde{A} = \mbox{Diag}\big( \vec{\xi} \big) \cdot A  \cdot \mbox{Diag}\big( \vec{\xi} \big),
\end{equation} for $\mbox{Diag}$ the function that takes a vector and creates a square diagonal matrix. We call $\Tilde{A}$ the \textbf{maximal entropy adjacency} matrix (MEA). Then the sum of the $i^{th}$ row is $\sum_{j}\xi_{i}\xi_{j}a_{ij} = \xi_{i} \sum_{j} \xi_{j}a_{ij} = \chi \xi^2_{i}$. Hence, the probability transition matrix $M$ has the form $p_{ij} = a_{ij} \frac{\xi_i \xi_j}{ \chi \xi_i^2} = a_{ij} \frac{ \xi_j}{ \chi \xi_i}$.

The MERW algorithm has been shown to be quite powerful in applications of community detection, and have been shown to unify measures of centrality, \cite{burda2009localization,sinatra2011maximal,ochab2012maximal}. 

Given the foundation of mathematically representing graphs, one may then ask questions about the characteristics. For our purposes, the question on which nodes are anomalous and to what extend are these nodes outliers. Outlier detection in graphs is an extensively researched field that spans utilizing random walks \cite{burda2009localization,chakrabarti2004autopart,sensarma2015survey,wang2018new,wang2019applying} to analyzing the structure of the graph \cite{francisquini2022community} and to leveraging the neural networks to extract information embedded in the graph \cite{zhao2022graph}. 

While not all classical methods are not directly applicable to a quantum processor, the techniques provide intuition. Particularly, the algorithm describe in Wang et al. \cite{wang2018new} that utilized the normalized adjacency matrix as the stochastic matrix to traverse the graph and scores nodes with respect of the probability of visiting the node.

\section{Continuous Time Quantum Walk}\label{sec:cont}

The area of continuous random walks on graphs (CTRW) is an extremely comprehensive subfield of quantum algorithms \cite{aharonov2001quantum,childs2002example,ahmadi2003mixing,solenov2006continuous,kendon2011perfect,bai2013quantum,wong2016laplacian,izaac2017centrality,benedetti2019continuous,wang2020experimental,tanaka2022spatial,goldsmith2023link,brown2023continuous}, established by Aharonov et al. \cite{aharonov2001quantum}. The basis for the walks is an approximation of the continuous Hamiltonian, which is typically created by the Laplacian of a graph, derived in Equation \ref{eq:lapc}. 

While CTRW has numerous applications \cite{wang2020experimental,malmi2022spatial,tanaka2022spatial,goldsmith2023link}, currently, there are no algorithms to score potential outlier nodes. For completeness, there are, however, algorithms for anomaly detection in general \cite{baker2022quantum,guo2023quantum,wang2023quantum}. 

In this section we give a thorough derivation of different Hamiltonians of a graph consisting of the Laplacian, MEA, and adjacency matrices, and the similarities and differences between them. We then derive the algorithm to anomaly score nodes with a continuous time random walk with the adjacency matrix at the core of the graph traversing. 

\subsection{Laplacian Matrix and Hamiltonian}\label{subsec:lap}
For implementation of a continuous walk in a quantum circuit, we require a Hamiltonian that captures the dynamic of the graph and a time-step to approximate the real-time evolution. For the given time-step $\gamma$, number of steps $t$, and Hermitian matrix $M$, define the Hamiltonian $H_B(\gamma,t) := \exp\left\{ \minus i\gamma t \cdot B \right\}$, where throughout the manuscript we set Planck's constant to $1$, $\hbar=1$.

For the Laplacian, $L$, we now analyze the Hamiltonian $H_L(\gamma,t)$. Observe that, for a given operator norm $\mathfrak{N}$, the difference between $H_L(\gamma,t)$ and $H_{\minus A}(t)$ is almost of the form $e^{Y+X} - e^X$, which $\mathfrak{N}\Big( e^{Y+X} - e^X \Big) \leq \mathfrak{N}(Y)e^{\mathfrak{N}(Y)} e^{\mathfrak{N}(X)}$.  However, the difference of the $-1$ scalar adds an extra term. Particularly, one may show that 
\begin{equation}\label{eq:op-norm}
\begin{split}
    & \mathfrak{N}\Big( H_L(\gamma,t) - H_{ \minus A}(\gamma,t) \Big)  \\ 
    & \leq \mathfrak{N}(-i\gamma t \cdot D)e^{\mathfrak{N}(-i\gamma t \cdot D)} e^{\mathfrak{N}(i\gamma t \cdot A)} + 2\cdot\mathfrak{N}\Big(\sin(i\gamma t \cdot A) \Big).
\end{split}
\end{equation}
Thus, the complexity of the graph dictates how close the Hamiltonian with the adjacency matrix approximates the Laplacian Hamiltonian. However, the evolution of each Hamiltonian may completely differ from each other \cite{wong2016laplacian}.

For the MEA matrix, observe by Equation \ref{eq:mea}, for all $i$ and $j$ we have $|a_{ij}\xi_{i}\xi_{j} - a_{ij}|=|\xi_{i}\xi_{j} -1 | \cdot |a_{ij}|$. Hence, the closer the eigenvector of the largest eigenvalue is to the vector of all $1's$ the closer the MEA and adjacency matrix are equivalent. This scenario would happen when all of the nodes are connected, which would imply the Laplacian is also close to the adjacency matrix by Equation \ref{eq:op-norm}. More generally, the magnitude of the eigenvector of the largest eigenvalue corresponds the the norm of the matrix $D$ in Equation \ref{eq:lapc}. Therefore, the Laplacian, adjacency, and MEA matrices are intricately connected, but differ enough to influence traversing the graph in distinguishable ways. 

\subsection{Mixing}\label{subsec:mix}
Before introducing the anomaly scoring method we establish a mathematical foundation that, given initial state, the average of the expected states converge to a distribution of states. Recall from Section \ref{sec:intro} we defined a vertex to be an outlier if it is not traversed often. 

This subsection follows the work in Aharonov et al. \cite{aharonov2001quantum}, which is the first research on quantum walks on graphs and established the foundation of quantum walks on graphs. 

Taking the standard notation of the state $\ket{\psi_0}$ as the initial state, $U$ as the unitary operator, and $\ket{\psi_t} = U^t\ket{\psi_0}$, we define $P_t(\phi|\psi_0) := |\braket{\phi}{\psi_t}|^2$ as the probability of state $\ket{\phi}$ at time $t$, and 
\begin{equation}\label{eq:avg-prob}
\displaystyle \overline{P}_{t_n}(\phi|\psi_0) := \frac{1}{t_n}\sum_{t=0}^{t_n-1}P_t(\phi|\psi_0)    
\end{equation}
as the average of state $\ket{\phi}$ over the time period $[0,t_n \minus 1]$. 

\begin{re}
Observe that Equation \ref{eq:avg-prob} is necessary since an invariant state of the process does not exist. To understand why this true we assume there does exist an invariant state $\ket{\pi}$. Then we have
\begin{equation}\label{eq:no-invar}
    \begin{split}
        \lim_{t \to \infty} &  \| \ket{\psi_{t+1}} - \ket{\psi_t} \|  =  \| U\ket{\pi} - \ket{\pi} \| \\ 
        & = \| \ket{\pi} - \ket{\pi} \| = 0,
    \end{split}
\end{equation}
where the first equality comes from the equality $\| \ket{\psi_{t+1}} - \ket{\psi_t} \| = \| U(\ket{\psi_{t}} - \ket{\psi_{t-1} } ) \| = \| \ket{\psi_{t}} - \ket{\psi_{t-1} } \|$ which holds for all $t$ since $U$ is measure preserving. Thus, we conclude that if $\ket{\pi}$ exists then all states are the same which implies $U$ is the identity operator.  
\end{re}

Theorem 3.4 in Aharonov et al. \cite{aharonov2001quantum} states that if the operator $U$ is unitary then $\displaystyle \lim_{t_n \to \infty} \overline{P}_{t_n}(\phi|\psi_0)$ exists and is unique with respect to the initial condition $\ket{\psi_0}$. The authors denote this limit as the \textbf{limiting distribution}. Define $\displaystyle \pi(\cdot| \psi_0) := \lim_{t_n \to \infty} \overline{P}_{t_n}(\cdot|\psi_0)$. 

Theorem 3.4 in Aharonov et al. yields the following proposition. 
\begin{prop}\label{prop:avg-prob}
    If $G$ is a symmetric weighted graph then the limiting distribution of the Hamiltonian of the Laplacian, adjacency, or MEA matrix exists and is unique with respect to the initial condition. 
\end{prop}

To understand the convergence of the limiting distribution, define the \textbf{mixing time} of a quantum Markov process as $M_{\epsilon} = \min\big\{ \Tilde{t} \big| \forall t \geq \Tilde{t}, \ket{\phi_0} : || \pi(\cdot|\phi_0) -  \overline{P}_{t}(\cdot|\psi_0) || < \epsilon \big\}$. The mixing time gives a rigorous definition of calculating how many time-steps it takes to be close to the limiting distribution. Given infinite possible number of initial conditions $\ket{\psi_0}$, we take $\ket{\psi_0}$ to be a basis state. 

Adjusting for the global perspective of the mixing time, we will define the \textbf{sampling time}, given by $S_{\epsilon}$, to give a point-wise calculation to be close to a subset of $V$. Ergo, $S_{\epsilon} = \min\big\{ \Tilde{t} \big| \forall t \geq \Tilde{t}, \Tilde{t}, \ket{\phi_0}, X \subseteq V : | \pi(X|\phi_0) - \overline{P}_{t}(X|\psi_0)| < \epsilon \pi(X|\phi_0) \big\}$.

\section{Symmetric Graph Anomaly Scoring}\label{sec:sym-score}

With the given rigorous foundations given with Equation \ref{eq:avg-prob} and Proposition \ref{prop:avg-prob} we are now ready to describe the outlier scoring algorithm. At a high-level, the algorithm traverses the graph with respect to Equation \ref{eq:avg-prob} to identify well-traversed vertices, and after a sufficient number of runs the multiplicative inverse of the probability of a state is then its anomaly score. Granular description and logic of the derivation is given within this subsection. 

For the CTRW to score nodes on being an outlier, we posit the adjacency matrix invokes a higher entropy over the Laplacian, since the energy flow of the graph is perfectly described the Laplacian, ensuring the entire graph is consistently traversed. However, the adjacency matrix is not too biased, unlike the MEA that would identify the connected subgroup of the graph that is central to traversing, but may put over emphasis on this subset dictated by this centrality \cite{ochab2012maximal}. 

Observe that, as noted in Wong and Lockhart \cite{wong2021equivalent}, the quantum walk with the Laplacian or adjacency matrix is the same when a graph is regular (when each node has the same number of neighbors). Wong and Lockhart further study irregular graphs and show when the initial condition is localized at a certain vertex, the distributions of walks converge to the same probability. 

The centrality from the adjacency matrix is described by Wang et al. \cite{wang2022continuous}. This is further emphasized by Brown et al. \cite{brown2023continuous} as the authors consider the Hamiltonian of the adjacency matrix for the evolutionary operator of the CTRW and show relationships of state transfer between to two disjoint subsets of vertices.

\begin{algorithm}
\textbf{Input:} Symmetric adjacency matrix $A$ of size $2^n \times 2^n$, time-step $\gamma$, number of steps $t$, number of walks $w$, and empty array $A_r$ of size $T \times 2^n$
\begin{algorithmic}[1]
\State{$U \gets \exp\{ -i\cdot \gamma \cdot A \}$}
\Procedure{CTRWAnomalyScore}{$U,t,w,A_r$}
\State{$\displaystyle \ket{\vec{P}_w} \gets \frac{1}{\sqrt{2^n}}\sum_{i=0}^{2^n-1} \ket{i}$}
\For{$i \gets 1, t$}
    \State{$\ket{\psi_{0}} \gets \ket{\vec{P}_w}$}
    \State{$k \gets i$}
    \While{$k > 0$}
        \If{$\mbox{floor}(k/w) > 0$}
        \State{$k' \gets w$}
        \Else
            \State{$k' \gets k \mod w$}
        \EndIf
        \State{$\ket{\psi_{i'}} \gets U^{k'} \ket{\psi_{i'}}$}
        \State{$k \gets k - k'$}
    \EndWhile
    \State{$\vec{P}_w \gets \mbox{frequencies}( \ket{\psi_{i}})$}
    \Comment{probability vector}
    \State{$\displaystyle \ket{\vec{P}_w} \gets \sum_{i=0}^{2^n-1} \sqrt{ \vec{P}_w[i] } \cdot \ket{i}$}
    \State{$A_r[i] \gets \vec{P}_w$} 
    \Comment{append the $i^{th}$ row}
\EndFor
\State{ $\displaystyle \vec{P}_t \gets \frac{1}{t} \sum_{i=0}^{t-1} A_r[i]$ }
\State{\textbf{return} $\vec{P}_T^{-1}$ } \Comment{inverse is applied entry-wise}
\EndProcedure
\end{algorithmic}
\caption{Symmetric CTRW Anomaly Score}\label{alg:ctw}
\end{algorithm}

\begin{re}
While discrete and continous quantum walks have been shown to yield a theoretical speed-up over the classical counterpart under certain conditions \cite{szegedy2004quantum,varsamis2022hitting, PhysRevLett.129.160502}, the actual implementation on hardware devices can be challenging even with the augmentation of the algorithm with an intermediate data loading technique. However, it is expected that such a theoretical speedup might be eventually realized with fault-tolerant quantum computers \cite{freedman2014faulttolerantquantumwalks}.
\end{re}

Observe that the basis for Algorithm \ref{alg:ctw} yields similar form to Equation \ref{eq:avg-prob}. The adjustment stems from the need of crafting a circuit with reasonable depth. However, as we will show, the two are equivalent in distribution. Lastly, starting the CTRW with every node in equal superposition is necessary as we assume there is no knowledge a priori.

\begin{claim}\label{claim:algo}
Equation \ref{eq:avg-prob} and the finally probability given in Algorithm \ref{alg:ctw} are equivalent in distribution. 
\end{claim}
\begin{proof}
Given the Markovian nature of the operator, it is sufficient to show the equivalence holds for an arbitrary positive integers $t$ and $w$, initial condition $\ket{\psi}$.

The Spectral Decomposition theorem (see Theorem 2.1 in \cite{nielsen2010quantum}) yields that $U$ is diagonalizable, ergo, $U = T^{\dag} \Lambda T$ where $T$ is unitary where each column a basis vector of an orthonormal basis and $\Lambda$ is a diagonal matrix. It is straightforward to derive the equality $U^n = T^{\dag} \Lambda^n T$. Thus, the frequency to be in the a basis state after apply this operator is 
\begin{equation}\label{eq:expt}
\begin{split}
 & \bra{\psi}(U^n)^{\dag} U^n \ket{\psi} = \bra{\psi} (T^{\dag} \Lambda^n T)^{\dag} T^{\dag} \Lambda^n T \ket{\psi} \\ 
 & = \bra{\psi T} \Lambda^{2n}\ket{T\psi} = \sum_{i} \bra{\psi T}\Big( \lambda^{2n} \ket{i}\bra{i} \Big) \ket{T\psi} \\
 & = \sum_{i} \lambda_i^{2n} |\braket{i}{T\psi}|^2.    
\end{split} 
\end{equation}

Proof by induction will show equality to Equation \ref{eq:expt}. For the case $n=n_0+n_1$ for $n_0, n_1 \in \mathbb{N} \setminus \{0\}$ we apply $U^{n_0}$ to the initial condition but not simplify with respect to Equation \ref{eq:expt}, which yields $\ket{U^{n_0}\psi}$.

Now, we apply the operator $U^{n_1}$ to the state $\ket{U^{n_0}\psi}$ and take expectations, which gives

\begin{equation*}
    \begin{split}
     & \bra{U^{n_0}\psi} \big(U^{n_1} \big)^{\dagger}  U^{n_1} \ket{U^{n_0}\psi} \\ 
     & = \bra{\psi} \big( U^{n_0} \big)^{\dagger} \big(U^{n_1}\big)^{\dagger}  U^{n_1}U^{n_0} \ket{\psi} \\
     & = \bra{\psi} (U^{n})^{\dagger}  U^{n}\ket{\psi} \\
     & = \sum_{i} \lambda_i^{2n} |\braket{i}{T\psi}|^2,
    \end{split}
\end{equation*}
with the last equality following the logic in Equation \ref{eq:expt}. The case when we assume $\displaystyle \sum_{i=0}^{t-1}n_i$ holds follows exactly as above.
\end{proof}

For the question on the size of the time-step $\gamma$, we will display that if $\gamma$ is sufficiently small then any value below the threshold can be used with an extremely similar convergence. To show this set a threshold $\epsilon>0$, $\gamma_1<\gamma_0<\epsilon/2$, and an arbitrary $t$ and Hermitian matrix $B$. 
\begin{equation}\label{eq:gamma-diff}
    \begin{split}
    & \mathfrak{N}\Big( H_B(\gamma_0,t) - H_B(\gamma_1,t) \Big)  = \mathfrak{N}\left( \sum_{n=1}^{\infty} \frac{(-it)^nB^n(\gamma_0^n - \gamma_1^n)}{n!} \right) \\
    & = \mathfrak{N}\left( \sum_{n=1}^{\infty} \frac{(-it)^nB^n (\gamma_0 + \gamma_1)^n}{n!}\cdot \frac{\gamma_0 - \gamma_1}{\gamma_0 + \gamma_1} \right) \\
    & = \mathfrak{N}\left( \sum_{n=1}^{\infty} \frac{(-it)^nB^n (\gamma_0 + \gamma_1)^n}{n!}\cdot \frac{\gamma_0 - \gamma_1}{\gamma_0 + \gamma_1} \right) \\
    & = \mathfrak{N}\left( \frac{\gamma_0 - \gamma_1}{\gamma_0 + \gamma_1} (-it)B(\gamma_0 + \gamma_1) \sum_{k=0}^{\infty} \frac{(-it)^kB^k (\gamma_0 + \gamma_1)^k}{(k+1)!}  \right) \\
    & = \mathfrak{N}\left( \frac{\gamma_0 - \gamma_1}{\gamma_0 + \gamma_1} (-it)B(\gamma_0 + \gamma_1) \right. \\
    & \hspace{25pt} \left. \cdot \sum_{k=0}^{\infty} \frac{(-it)^kB^k (\gamma_0 + \gamma_1)^k}{(k)!}\cdot \frac{1}{(k+1)^{1/k}}  \right) \\
    & \leq t(\gamma_0 - \gamma_1)\cdot\mathfrak{N}\left( B \right) \mathfrak{N}\left( H_B(\epsilon,t)  \right).
    \end{split}
\end{equation}
Therefore, the norm of the matrix $B$ dictates the size of $\epsilon$, which is intuitive since the larger the norm of $B$ the more dynamic the energy flow of the graph.  

\begin{re}
Observe that the adjacency matrix is not necessarily unique in assuring the convergence of the algorithm, any symmetric matrix the captures the dynamic of nodes and edges of the of a graph can be implemented in Algorithm \ref{alg:ctw}. As noted in Childs et al. \cite{childs2002example}, as long as the Hermitian matrix has structure, the Hamiltonian represents a random walk. However, as will be argued and shown numerically, the adjacency matrix naturally captures structure of the graph while not creating bias. 
\end{re}

\subsection{Asymmetric Graph Anomaly Scoring}\label{subsec:asym-score}
For the case of a directed graphs, if the adjacency matrix is symmetric, as is possible with cyclic graph, then Algorithm \ref{alg:ctw} may be applied. Note that if the adjacency matrix would be symmetric if not for the weights (ergo $e(i,j), e(j,i)\in E$ but $e(i,j) \neq e(j,i)$) then take the average and replace this new value. Ergo, create apply the adjusted adjacency matrix $\Tilde{A} = \Big( (A+A^T) + (A+A^T)^T\Big)/2$. The average encompasses the energy flow between respective nodes.

This logic can be extended to any adjacency matrix of a directed graph by leveraging the derivation of the Hermitian-adjacency matrix \cite{guo2017hermitian,mohar2020new,kubota2021quantum,yuan2022hermitian,chaves2023and}. Hermitian-adjacency matrix were derived by Gou and Mohar \cite{guo2017hermitian}, with the idea of not only creating a Hermitian out of an asymmetric real-valued matrix, but further captures the deeper dynamic of an unweighted directed graph by including central nodes that only have out-going edges. However, the adjustment does not directly give missing edges largely impactful weights, but are replaced with either the imaginary number $i$ or with $-i$, with respect to if the node is the \textit{direct successor} (edge flowing to) or the \textit{direct predecessor} (edge flowing from). 

We take the derivation given by Mohar \cite{mohar2020new} to calculate the Hermitian-adjacency matrix. Define $\alpha = a + bi$ where $|\alpha|=1$ and $a \geq 0$. The Hermitian adjacency matrix $A^{\alpha}$ has the entries $a^{\alpha} = $ if 
\begin{eqnarray}\label{eq:herm-adj}
a^{\alpha}_{ij} = \left\{
\begin{array}{ll}
      a_{ij}   & \mbox{if } a_{ij} = a_{ji}  \\
       a_{ij}\alpha + a_{ji}\overline{\alpha}  & \mbox{otherwise}
\end{array} 
\right.
\end{eqnarray}

Define the function HermAdj$(A,\alpha)$ to be the function that maps an adjacency matrix $A$ to $A^{\alpha}$ with entries map exactly as in Equation \ref{eq:herm-adj}. We are now ready to state the algorithm. 

\begin{algorithm}
\textbf{Input:} Asymmetric Adjacency matrix $A$ of size $2^n \times 2^n$, complex number to calculate Hermitian matrix $\alpha$, time-step $\gamma$, number of steps $t$, number of walks $w$, and empty array $A_r$ of size $t \times 2^n$
\begin{algorithmic}[1]
\State{$A^{\alpha} \gets$ HermAdj$(A,\alpha)$}
\State{$U \gets \exp\{ -i\cdot \gamma \cdot A^{\alpha} \}$}
\State{\textbf{return} CTRWAnomalyScore$(U,t,w,A_r)$}
\end{algorithmic}
\caption{Asymmetric CTRW Anomaly Score}\label{alg:ctw-asym}
\end{algorithm}

\section{Discrete Time Quantum Walk}\label{sec:coin}
We now consider discrete time quantum random walks on graphs with respect to Algorithm \ref{alg:ctw} and show little work is required to exchange this subclass of algorithms into the previous arguments. Unlike continuous quantum random walks, discrete time quantum random walks is a diverse subfield with many algorithms. See Kadian et al. \cite{kadian2021quantum} for a comprehensive review of discrete and continuous time quantum walks, and Childs \cite{childs2010relationship} and Kendon \cite{kendon2006quantum} for a detailed description of the differences between the two methods. 

Discrete time quantum walks have four essential forms: Szegedy \cite{szegedy2004quantum}, which rigorously derived a quantum algorithm based Markov chain walk; coined walk \cite{aharonov2001quantum} that splits the circuit into a coin space and state space, with the coin space dictating the next step of the walk; staggered quantum walks \cite{portugal2016staggered} which is a subclass of coined quantum walks and was derived to traverse a graph via partitions; and scattering \cite{kendon2007decoherence} that emphasized the important effect of decoherence. Given the large intersection of the four subclasses of discrete time quantum walk algorithms, our focus will be on the generalized coined walk.

A coined quantum walk consists of a state space $\mathcal{H}_p$, and coined space $\mathcal{H}_c$, and the quantum algorithm spans the space $\mathcal{H}_p \otimes \mathcal{H}_c$. With this algorithm there are two overarching operators, the \textbf{coin operator} and \textbf{shift operator}. The coin operator determines the next direction of the walk and has the form $C = \displaystyle \sum_{i} \ket{i}\bra{i} \otimes C_i$, where $\ket{i}$ is a state in $\mathcal{H}_p$ and $C_i$ is the local operator representing all adjacent vertices of $i$ and is typically the Grover diffusion coin operator \cite{moore2002quantum}. This operator is denoted as inhomogeneous and notice that if $C' := C_i = C_j$ for all $i,j$, then $C= \mathbbm{1}_p \otimes C'$ and is denoted as homogeneous. 

The shift operator works off of the current state and the output of the coin operator to then calculate the next state of the walk. The general the shift operator is local and incorporates the nodes of a vertex but does not affect the coin space. With the degree of node $i$ denoted as $d_i$ and $v(i,j)$ denoted as the neighbor of $i$ connected with the $j^{th}$ edge, the shift operator has the form $\displaystyle S = \sum_{i}\sum_{j=1}^{d_i} \ket{v(i,j)}\bra{i} \otimes \ket{j}\bra{j}$. 

The final operator for the discrete time walk is $U = SC$ and if $U$ is unitary then Aharonov et al. \cite{aharonov2001quantum} showed the same results in the Mixing Subsection \ref{subsec:mix}, and Claim \ref{claim:algo} also holds. From this basis, one may apply Algorithm \ref{alg:ctw} in a similar manner, except only measuring the quantum register representing the state space $\mathcal{H}_p$.   

\subsection{Weighted Graphs}\label{subsec:disc-weighted}
The derivation of the discrete time quantum walk operator $U=SC$ was given for an unweighted graph or an unweighted \textbf{reversible} directed graph. The term reversible refers to adjacent vertices $(v_i,v_j)$ which if there is an edge from $v_i$ to $v_j$ then there is an edge from $v_j$ to $v_i$; see Montanaro \cite{montanaro2005quantum} for further information. 

For a weighted graph we follow the method in Wong \cite{wong2017coined}, where Wong uses the foundation of Szegedy \cite{szegedy2004quantum}. For $A$ the adjacency matrix and node $v_i$, defining
\begin{equation}
    \ket{s_{v_i}} = \frac{1}{ \sqrt{\sum_j a_{ij}} }  \sum_j \sqrt{a_{ij}} \ket{j}
\end{equation}
then 
\begin{equation}
    C_i = 2\ket{s_{v_i}}\bra{s_{v_i}} - \mathbbm{1}_c.  
\end{equation}

\onecolumngrid

\begin{figure}[!ht]
    \begin{subfigure}[b]{0.45\textwidth}
    \includegraphics[width=.95\textwidth]{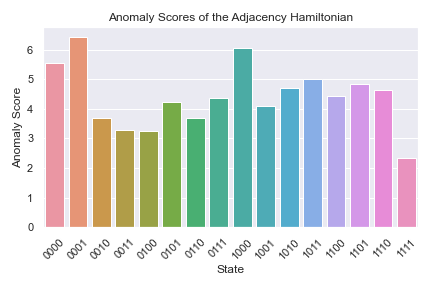}\caption{Adjacency matrix Hamiltonian anomaly scores for every node.}\label{subfig:anom-adj}
     \end{subfigure}
     \hfill
     \begin{subfigure}[b]{0.45\textwidth}
    \includegraphics[width=.95\textwidth]{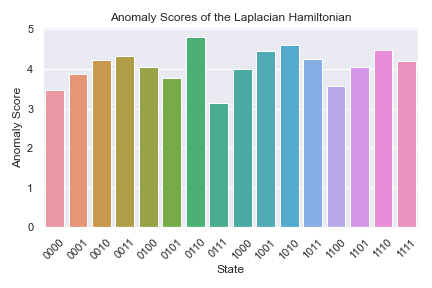}\caption{Laplacian matrix Hamiltonian anomaly scores for every node.} \label{subfig:anom-lap}
     \end{subfigure}
     \hfill
     \begin{subfigure}[b]{0.45\textwidth}
    \includegraphics[width=.95\textwidth]{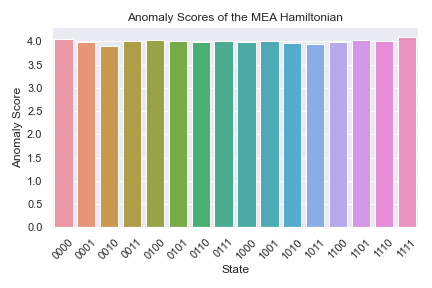}\caption{MEA matrix Hamiltonian anomaly scores for every node.}\label{subfig:anom-mea}
     \end{subfigure}
     \hfill
     \begin{subfigure}[b]{0.45\textwidth}
    \includegraphics[width=.95\textwidth]{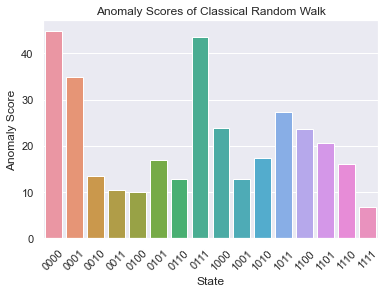}\caption{ Classical random walk anomaly scores for every node.}\label{subfig:anom-clsscl}
     \end{subfigure}
\caption{Anomaly scores for the adjacency, Laplacian, and MEA matrices, and a classical random walk. The node labels are converted to binary. the for the classical walk the damping term has value 0 and there is an error threshold of .05.}
\label{fig:anom-scores}
\end{figure}

\twocolumngrid

In general, there are graph instances where the walk operator $U= S(\mathbbm{1} \otimes C)$ is not unitary. With this case, if $U$ is Hermitian then there are method to transform the operator to be unitary \cite{daskin2017ancilla,schlimgen2022quantum}. Schlimgen et al. \cite{schlimgen2022quantum} display a straightforward implementation with the new operator
\begin{equation}
    \hat{U} = \left( \begin{array}{cc}
        U & \sqrt{\mathbbm{1} - U U^{\dagger}}  \\
        \sqrt{\mathbbm{1} - U U^{\dagger}} & U^{\dagger}
    \end{array} \right).
\end{equation}
The output from the circuit with $\hat{U}$ to the input vector $\ket{0} \otimes \ket{in}$, where $\ket{in}$ is the desired initial state, has the form
\begin{equation}
    \hat{U}\ket{0} \otimes \ket{in} = \left( \begin{array}{c}
        U \ket{in}\\
        \sqrt{\mathbbm{1} - U U^{\dagger}} \ket{in} 
    \end{array} \right).
\end{equation}
After applying the operator $\ket{0} \bra{0}\otimes \mathbbm{1}$ classically to the output vector collapses to the state $U\ket{in}$ \cite{daskin2017ancilla}.

\section{Experiments}\label{sec:diff-ham}

\subsection{Symmetric Graph with Different Hamiltonians and Classical Random Walk}

\begin{figure}[!ht]
    \centering
    \includegraphics[width=270px]{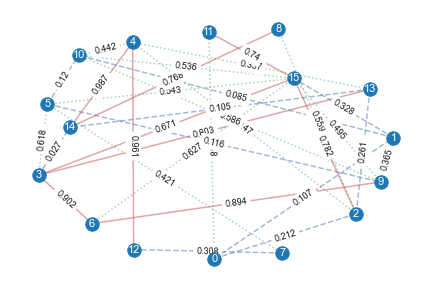}
    \caption{The undirected graph that is randomly generated with weights randomly generated and are within the unit interval. To assist in visualizing the graph, the edges of the vertices are split into three groups, values less than $1/3$ are identified with the blue dashed edge, values greater than $1/3$ and less than $2/3$ are identified by the green dotted edges, and values greater than $2/3$ are identified by the red solid line.  }
    \label{fig:undir-grph}
\end{figure}

To display the importance of using the adjacency matrix as the basis for Algorithm \ref{alg:ctw} we create a random weighted graph and run the algorithm with the Hamiltonian created from the adjacency, Laplacian, and MEA matrices. 

The undirected graph was randomly created with the Networkx \cite{hagberg2008exploring} Python package, and the weights of the edges were randomly generated with the Numpy \cite{harris2020array} Python package using the uniform distribution method with the interval $(0,1)$. A visualization of the graph is given in Figure \ref{fig:undir-grph}.

For the experiments the Qiskit Python package was used for the quantum circuits with the parameters of steps $t=40$, number of walks $w=4$, and time-step $\gamma = 1/(2\cdot\sqrt{13})$. For each step in the algorithm the circuit was sampled with $30000$ shots to approximate the expected outcome of the circuit. 

The Grover-Rudolph procedure \cite{grover2002creating} was used for each restart of the iteration of the walks. Interestingly, the number of gates for this procedure is considerable less than the walk operator, hinting at a small increase to the complex of a quantum walk, keeping the quantum advantage in certain classes of graphs. Particularly, the walk operator had a total number of 492 gates ($CX$-242, $R_Z$-84,$U_2$-74, $H$-64, $R_{zx}$-12,$R_{XX}$-6,$R_{YY}$-6, and $R_X$-4), and the Grover-Rudolph procedure had a total of 161 gates ($CX$-56, $CU$-70,$X$-34,and $R_Y$-1). Of course, in applications, the native gates of the processor will dictate these numbers.

The classical random walk for anomaly scoring similarly follows the algorithm Moonesinghe and Tan \cite{moonesinghe2006outlier}, with the variation of using the probability matrix $AD^{-1}$ instead of the normalized similarity-matrix, and there is the extra step of normalizing the vector with the recursive scores. Moonesinghe and Tan's algorithm utilizes the invariance of the Markov process, where the invariant probability measure $\pi$ with the recursion $\displaystyle \pi^{t+1} = \frac{ P^T \pi^{t} }{ || P^T \pi^{t} ||_1}$, where $||\cdot||_1$ is the $L^1$ norm. The damping term, $d$, is an adjustment to $\pi^{t+1}$,$\displaystyle \tilde{\pi}^{t+1} = \frac{d + (1-d) P^T \tilde{\pi}^{t} }{|| d + (1-d) P^T \tilde{\pi}^{t} ||_1}$, in order to force convergence. While not necessary for the experiments, from the calculations, the dampening term was shown to affect the scores, shrinking the distance between all of the scores.

\begin{table}[ht]
\centering
\begin{tabular}{|l | l | l | l |l |} \hline
Hermitian Matrices & Adjacency & Laplacian & MEA & Classical \\
\hline
Adjacency & .0    & .1840  & .1364 & .0780 \\ \hline
Laplacian & .1840  & .0    & .0241 & .3056 \\ \hline
MEA       & .1364 & .0241 & .0    & .1404 \\ \hline
Classical & .0780 & .3056 & .1404 & .0 \\\hline
\end{tabular}
\caption{ Symmetric KL differences between each of the scores.} 
\label{tab:kl}
\end{table}

Figure \ref{subfig:anom-adj} gives the anomaly scores with the Hamiltonian generated with the adjacency matrix. The states of $0000$, $0001$, and $1000$ have the highest scores. The nodes $0000$ and $0001$ are intuitive since these states are, at best, weakly connected to main graph since the nodes these vertices are connected to have a mid-range anomaly scores. While the node $1000$ is connected to a main node, $0010$, this connection is weak. The Hamiltonian generated with the adjacency matrix and the classical random walk, as noted from the histogram figures in Figure \ref{subfig:anom-adj} and symmetric KL scores in Table \ref{tab:kl}, are extremely close, but differ with the scores of nodes $1000$ and $0111$. The close symmetric KL scores of the final probabilities of the quantum random walks, with the largest score $.1840$, indicate that the probabilities of the events (nodes) are relatively. Moreover, while this shows that while the evolution of each Hamiltonian are similar, the simple act of taking the inverse of each state yields different behavior.

For completeness, taking two probability distribution $P$ and $Q$ over state space $\mathcal{X}$, KL divergence $\mbox{KL}(\cdot || \cdot )$  is defined as $\displaystyle \mbox{KL}( P || Q ) := \sum_{x \in \mathcal{X}} P(x) \log \left( \frac{P(x)}{Q(x)}\right)$, and symmetric KL is defined 
$$
\mbox{symKL}(P,Q) := \frac{KL( P || Q ) + KL( Q || P )}{2}.
$$ 

\begin{figure}[!ht]
     \centering
     \begin{subfigure}[b]{0.4\textwidth}
    \centering
    \includegraphics[width=1.\textwidth]{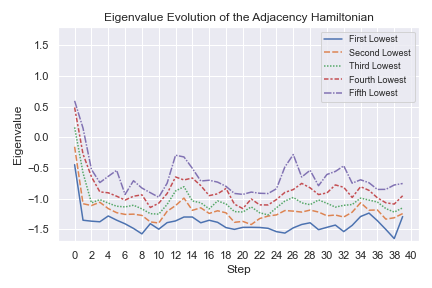}
     \end{subfigure}
     \hfill
     \begin{subfigure}[b]{0.4\textwidth}\label{fig:eng-adj}
    \centering
    \includegraphics[width=1.\textwidth]{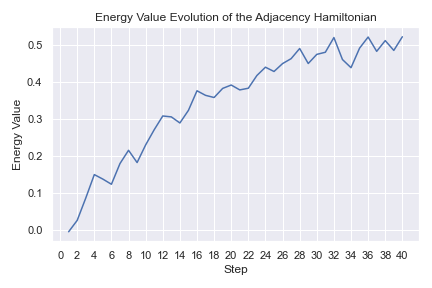}
     \end{subfigure}
\caption{Eigenvalues and total energy evolution for the adjacency Hamiltonian.}
\label{fig:behv-adj}
\end{figure}

Contrary to outliers, nodes $0100$ and $1111$ are clearly highly-connected nodes since they have an above average number of edges and each edge has a medium or large weight. These observations are captured in the scores of these two nodes with the adjacency matrix Hamiltonian. However, the scores of the nodes $0100$ and $1111$ are ambiguous with the scores generated with the Laplacian Hamiltonian. The scores generated with the MEA Hamiltonian are all ambiguous. Considering the derivation of the MEA matrix with the purpose of traversing the graph with bias, the Laplacian in the perspective of Hamiltonian oracles \cite{mochon2007hamiltonian,wong2020unstructured}, and starting the algorithm with each state in uniform superposition, the outlier scores from each Hamiltonian is intuitive since every vertex is consistently visited.

The anomaly scores are given in Figure \ref{fig:anom-scores}, and while scores of some nodes are similar (e.g., the states $1000$ and $1011$ ) there are nodes with scores that completely differ (e.g., $0000$, $0001$, $0010$, and $1111$ ). Interestingly, using symmetric KL divergence as the pseud-metric to compare, and displayed in Table \ref{tab:kl},  

\begin{figure}[!ht]
     \centering
     \begin{subfigure}[b]{0.4\textwidth}
    \centering
    \includegraphics[width=1.\textwidth]{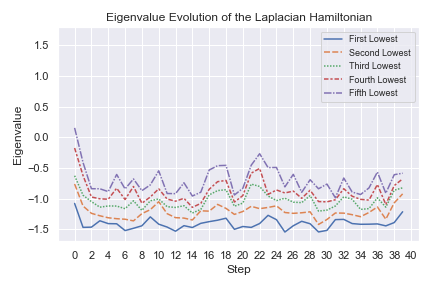}
     \end{subfigure}
     \hfill
     \begin{subfigure}[b]{0.4\textwidth}
    \centering
    \includegraphics[width=1.\textwidth]{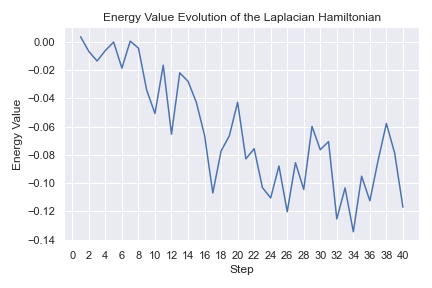}
     \end{subfigure}
\caption{Eigenvalues and total energy evolution for the Laplacian Hamiltonian.}
\label{fig:behv-lap}
\end{figure}

To further examine the behavior of each Hamiltonian at every step we calculated the five lowest eigenvalues and calculated the total energy (with respect to spin). Figure \ref{fig:behv-adj} displays the characteristics of the adjacency Hamiltonian, Figure \ref{fig:behv-lap} displays the characteristics of the Laplacian Hamiltonian, and Figure \ref{fig:behv-mea} displays the characteristics of the MEA Hamiltonian.

\begin{figure}[!ht]
     \centering
     \begin{subfigure}[b]{0.4\textwidth}
    \centering
    \includegraphics[width=1.\textwidth]{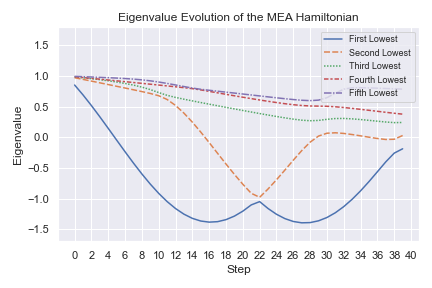}
     \end{subfigure}
     \hfill
     \begin{subfigure}[b]{0.4\textwidth}
    \centering
    \includegraphics[width=1.\textwidth]{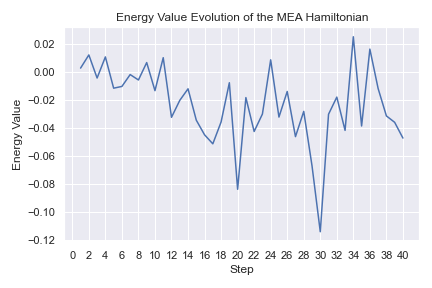}
     \end{subfigure}
\caption{Eigenvalues and total energy evolution for the MEA Hamiltonian.}
\label{fig:behv-mea}
\end{figure}

Observe that the eigenvalues for the adjacency Hamiltonian in have consistent small gaps between sequential eigenvalues. These small gaps indicate a high-entropy, which is reflected in energy of the system. The Laplacian Hamiltonian has, for the majority of the steps, larger gaps between sequential eigenvalues, and this bias is reflected in the total energy. Lastly, MEA Hamiltonian has a consistent tremendous gap between the lowest and second lowest eigenvalue, displaying the bias of the walk; this bias is further emphasized in the anomaly scores in Figure \ref{subfig:anom-mea}. Interestingly, this gap closes as time increases then immediately increases. In all, the consistent increase in energy to an upper bound of the system dictated by the adjacency Hamiltonian consistently, as backed by the gap between the eigenvalues, yields the desired behavior of traversing the graph in a more unbiased manner.

\begin{figure}[!ht]
     \begin{subfigure}[b]{0.4\textwidth}
        \includegraphics[width=250px]{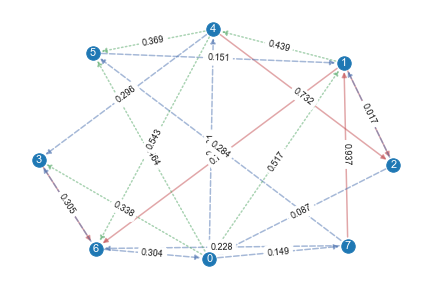}
    \caption{Directed graph with no apparent outlier vertices.}
    \label{subfig:dir-grph}
     \end{subfigure}
     \hfill
     \begin{subfigure}[b]{0.4\textwidth}
    \includegraphics[width=230px]{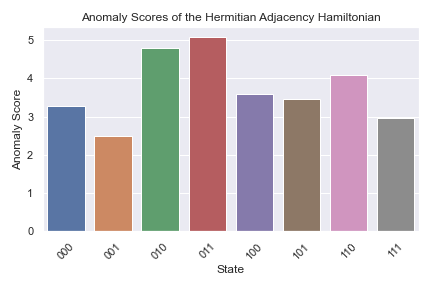}
    \caption{Anomaly scores from the algorithm for the directed graph above.}
    \label{subfig:anom-h-adj}
     \end{subfigure}
\caption{Example of a directed graph without an outlier node. The directed graph that is randomly generated with weights randomly generated and are within the unit interval. The edges of the vertices are denoted exactly as in Figure \ref{fig:undir-grph}. This graph has a high probability of having no high-value outlier nodes. }
\label{fig:dirgraph-no-anom}
\end{figure}

\begin{figure}[!ht]
     \begin{subfigure}[b]{0.4\textwidth}
        \includegraphics[width=230px]{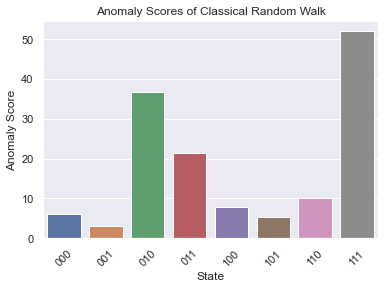}
    \caption{Anomaly scores from a random walk with damp term of 0 and error term of .05.}
    \label{fig:dir-grph}
     \end{subfigure}
     \centering
     \begin{subfigure}[b]{0.4\textwidth}
    \includegraphics[width=230px]{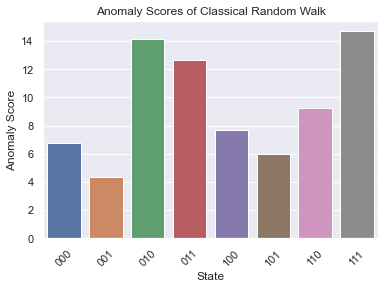}
    \caption{Anomaly scores from a random walk with damp term of .1 and error term of .05}
    \label{fig:anom-h-adj}
     \end{subfigure}
\caption{ The two bar graphs are the anomaly scores for the graph in Figure \ref{subfig:dir-grph}, with the same error of .05 but different dampening terms. The effect of the dampening term on the scores is evident, with apparent convergence to uniformity.} 
\label{fig:dirgraph-no-anom-clsscl}
\end{figure}

\begin{figure}[!ht]
     \centering
    \begin{subfigure}[b]{0.4\textwidth}
    \includegraphics[width=250px]{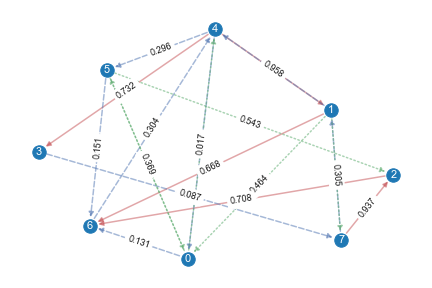}
    \caption{Directed graph with apparent outlier vertices.}
    \label{subfig:dir-grph-outlier}
     \end{subfigure}
     \centering
     \begin{subfigure}[b]{0.4\textwidth}
    \includegraphics[width=230px]{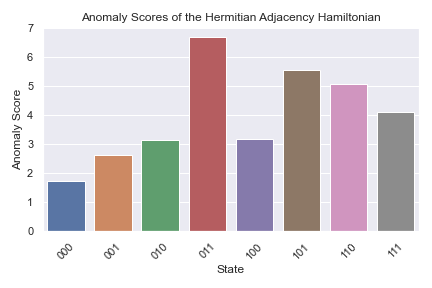}
    \caption{Anomaly scores for the directed graph above.}
    \label{subfig:asas_h-adj_anom}
     \end{subfigure}
\caption{Example of a directed graph with an outlier node. The directed graph that is randomly generated with weights randomly generated and are within the unit interval. The edges of the vertices are denoted exactly as in Figure \ref{fig:undir-grph}. This graph has a high probability of having no high-value outlier nodes. }
\label{fig:dirgraph-anom}
\end{figure}

\begin{figure}[!ht]
     \centering
     \begin{subfigure}[b]{0.4\textwidth}
    \includegraphics[width=230px]{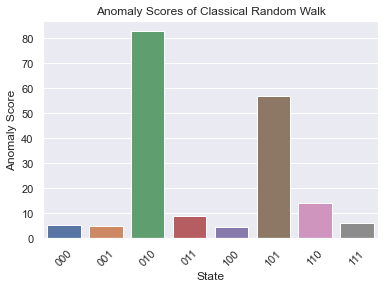}
    \caption{Anomaly scores from a random walk with damp term of 0 and error term of .05.}
    \label{subfig:asas_h-clsscl-0}
     \end{subfigure}
     \centering
     \begin{subfigure}[b]{0.4\textwidth}
    \includegraphics[width=230px]{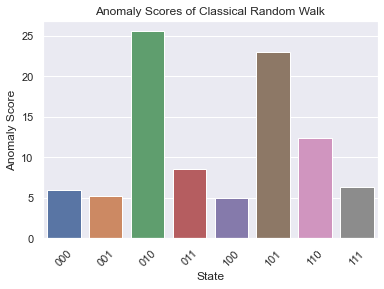}
    \caption{Anomaly scores from a random walk with damp term of .05 and error term of .05.}
    \label{subfig:asas_h-clsscl-5}
     \end{subfigure}
\caption{ The two bar graphs are the anomaly scores for the graph in Figure \ref{subfig:dir-grph-outlier}, with the same error of .05 but different dampening terms. The effect of the dampening term on the scores is evident, as shown in Figure \ref{fig:dirgraph-no-anom-clsscl}.} 
\label{fig:dirgraph-anom-clsscl}
\end{figure}

\subsection{Asymmetric Graphs with Adjacency Hamiltonian and Classical Random Walks}

To further illuminate this technique, this section has two small toy directed graphs, one with no clear outlier node and another with clear outlier nodes. Given the adjustment of the adjacency matrix via the Hermitian matrices transformation, Equation \ref{eq:herm-adj}, one would expected discrepancies between scores from Algorithm \ref{alg:ctw-asym} and classical random walks. In fact, this is observed. However, the dampening term may be leveraged to counter the bias in the classical walk, as shown in Figure \ref{fig:dirgraph-no-anom-clsscl} and Figure \ref{fig:dirgraph-anom-clsscl}.

For the first graph, consider the directed acyclic graph in Figure \ref{fig:dir-grph}. For the parameters, take $\alpha=i$, steps $t=20$, no walks, time-step $\gamma = 1/(2\cdot\sqrt{13})$, and for each step in the algorithm the circuit was ran $10000$ times. 

The outlier scores, given in Figure \ref{fig:anom-h-adj}, identify nodes $010$ and $011$ as fairly likely to be anomalous. However, given the small difference between the largest score and the smallest score, it is a safer assessment that none of the vertices are outliers. However, from the asymmetry of the graph, there is an unsurprising discrepancy between the scores. Particularly, while Algorithm \ref{alg:ctw-asym} yields scores close to uniformity, the classical walk with a zero dampening term, given in Figure \ref{subfig:asas_h-clsscl-0}, displays scores far less uniform, with nodes $010$ and $111$ identified as outliers. When a dampening term of $.1$ is included in the classical random walk, the scores are similar to Algorithm \ref{alg:ctw-asym}.

Now, consider the graph in Figure \ref{subfig:dir-grph-outlier}. From the visual inspection, it appears that nodes $010$, $011$, and $111$ are outliers. This intuition, however, is only supported for the node $011$, as indicated from the outlier scores in Figure \ref{fig:anom-h-adj}. Further inspection displays that node $101$ is connected to the heavily connected node $100$ and node $111$, where are in a four node loop. However, since node $111$ is not well-connected, one may see that node $011$ will not be visited often, that is, with the Hermitian adjacency matrix. Thus, the algorithm's determination is intuitive. Contrarily, if the graph is taken as is, since node $100$ is well-connected with only two outgoing edges, where the edge to $011$ has a large weight, this node will be visited often. This reasoning is displayed in the anomaly scores for the classical random walks given in Figure \ref{fig:dirgraph-anom-clsscl}. Furthermore, the classical random walks also identified nodes $010$ and $101$ as potential outliers, which makes superficial sense since both nodes only have two incoming edges, one of which is $111$ to $010$. These scores are quite large when the dampening value is $0$, and with a dampening value of $.05$, the anomaly scores for $010$ and $101$ are significantly decreased. However, since the anomaly score for $101$ is relatively large for all algorithms, this node has to be considered as an outlier.

\section{Discussion}\label{sec:discussion}

\subsection{Summary}

This paper derived an algorithm to calculate a score for each vertex in a graph. The algorithm was crafted to adjust for the NISQ, and subsequent fault-tolerant era, with the potential to negate the quantum advantage. However, the algorithm is easily adjusted for QPUs with high fidelity and scalable to the data. The evolution of the process within the algorithm was first given for a matrix representation of a graph that is Hermitian, then an explicit mitigation to the algorithm was derived when this representative matrix is not Hermitian. It was argued that that the adjacency matrix should be the representative matrix to utilize for traversing the graph as it balances variance and bias by capturing the not-too-biased entropy of the graph. However, the algorithm is robust enough to consider other representative matrices.

It was rigorously shown that the algorithm converges to the expected probability, where this probability is contingent on the initial condition. The initial condition given for the algorithm is the superposition of equal probabilities for each node, since the more appropriate distribution of the vertices is not known. Lastly, we discussed how to incorporate discrete time quantum walks into the respective algorithms, with a transformation of the walk operator when it is Hermitian and not unitary. 

With toy graphs, it was demonstrated that the algorithm scores nodes similarly to a classical walk on a symmetric graph, and is extremely selective on a directed graph. This selective characteristic of the algorithm has a potential advantage over classical random walks in data mining since the number of identified potential outliers is significantly decreased, resulting in quicker determination of new labeled data points. Moreover, the classical computation does not scale well, while the derived algorithm scales well with the use of a universal quantum processor.

\subsection{Future Work}

A potentially future direction of the research is to fully display the efficacy of the algorithm, a comparison of the algorithm against classical methods against real-world data with known, but unlabeled, outliers. Given the state of all quantum modalities, this data set needs to be small enough to be ran on either simulated or a real quantum processor. In addition, a more rigorous comparison with different graph types and structures will be studied to understand when CTQW will have an exponential/quadratic speedup over their classical counterpart with different implementation of the walkers.

While the authors believe this derivation of the algorithm is encompasses quite well a large number of classes of graphs, there are a few questions about the method where answering these question would warrant future research. The algorithm was discussed with the adjacency matrix, however, are there various classes of graphs when the discrete Laplacian matrix outperforms the adjacency matrix. For directed graphs, are there different transformations of an asymmetric adjacency matrix other than the Hermitian adjacency matrix?

The algorithm derived only concerns simple unweighted/weighted graphs where the edges, nodes, and subgraphs are, essentially, featureless. This observation begs the question on how to naturally incorporate features into the Hamiltonian in the algorithm. At first thought, one may add a symmetric matrices with the respective features to the adjacency matrix and apply the second order Suzuki-Trotter approximation as the Hamiltonian in the algorithm. However, we posit there is a deeper representation of the flow of information without utilizing machine learning compression techniques.

Of last note, the manuscript did not implement a discrete continuous walk for the algorithm since width and length of the circuit grows fairly quickly. Consequently, this restricts the comparison of performance with the algorithm implemented with continuous time quantum random walks and discrete time quantum walks on `interesting' graphs. However, it would be interesting to discover an instance of an interesting graph where discrete time out performs continuous time from either time to convergence or robustness of the scores for the vertices. The authors posit there are classes of graphs where this holds. 

%Finally, a use-case with respective data needs to identified to compare classical methods against the derived method. Performance should be judged on time to convergence, robustness of the scores for the vertices, and how the methods scale.  

\section{Disclaimer}
About Deloitte: Deloitte refers to one or more of Deloitte Touche Tohmatsu Limited (“DTTL”), its global network of member firms, and their related entities (collectively, the “Deloitte organization”). DTTL (also referred to as “Deloitte Global”) and each of its member firms and related entities are legally separate and independent entities, which cannot obligate or bind each other in respect of third parties. DTTL and each DTTL member firm and related entity is liable only for its own acts and omissions, and not those of each other. DTTL does not provide services to clients. Please see www.deloitte.com/about to learn more.

Deloitte is a leading global provider of audit and assurance, consulting, financial advisory, risk advisory, tax and related services. Our global network of member firms and related entities in more than 150 countries and territories (collectively, the “Deloitte organization”) serves four out of five Fortune Global 500® companies. Learn how Deloitte’s
approximately 330,000 people make an impact that matters at www.deloitte.com. 
This communication contains general information only, and none of Deloitte Touche Tohmatsu Limited (“DTTL”), its global network of member firms or their related entities (collectively, the “Deloitte organization”) is, by means of this communication, rendering professional advice or services. Before making any decision or taking any action that
may affect your finances or your business, you should consult a qualified professional adviser. No representations, warranties or undertakings (express or implied) are given as to the accuracy or completeness of the information in this communication, and none of DTTL, its member firms, related entities, employees or agents shall be liable or
responsible for any loss or damage whatsoever arising directly or indirectly in connection with any person relying on this communication. 
Copyright © 2024. For information contact Deloitte Global.

\bibliographystyle{unsrt}
\bibliography{qwbib}
\end{document}